\documentclass[12pt]{amsart}

\usepackage{amssymb,amsmath,amsfonts,nomencl,mathrsfs}
\usepackage{amsthm}  
\usepackage[arrow, matrix, curve]{xy}
\usepackage[latin1]{inputenc}
\newcommand{\IC}{\mathbb{C}}
\newcommand{\IR}{\mathbb{R}}

\newcommand{\ILL}{\mathscr{L}}

\newcommand{\IFF}{\mathscr{F}}

\newcommand{\IN}{\mathbb{N}}
\newcommand{\IZ}{\mathbb{Z}}

\newcommand*{\longhookrightarrow}%
               {\ensuremath{\lhook\joinrel\relbar\joinrel\rightarrow}}
\newcommand{\pa}{\slash\slash}

\newcommand{\Id}{{\rm d}}

\newcommand{\f}{\frac}
\newcommand{\nn}{\nonumber}

\newcommand{\sm}{\sim_b}

\theoremstyle{plain}            
\newtheorem{theorem}{theorem}[section]
\newtheorem{Lemma}[theorem]{Lemma}
\newtheorem{Corollary}[theorem]{Corollary}
\newtheorem{Theorem}[theorem]{Theorem}

\newtheorem{Propandef}[theorem]{Proposition and definition}

\theoremstyle{definition}       

\newtheorem{Remark}[theorem]{Remark}

\begin{document}

 \title[Semiclassical limits on infinite graphs]{Semiclassical limits of quantum partition functions on infinite graphs}

   \author{Batu G\"uneysu}
\address{Institut f\"ur Mathematik, Humboldt-Universit\"at zu Berlin, Germany. E-mail: gueneysu@math.hu-berlin.de}

\begin{abstract}
We prove that if $H$ denotes the operator corresponding to the canonical Dirichlet form on a possibly locally infinite weighted graph $(X,b,m)$, and if $v:X\to \mathbb{R}$ is such that $H+v/\hbar$ is well-defined as a form sum for all $\hbar >0$, then the quantum partition function $\mathrm{tr}(\mathrm{e}^{-\beta \hbar ( H + v/\hbar)})$ satisfies 
$$
\mathrm{tr}(\mathrm{e}^{-\beta \hbar ( H + v/\hbar)})\xrightarrow[]{\hbar\to 0+}\sum_{x\in X} \mathrm{e}^{-\beta v(x)} \>\text{  for all $\beta>0$},
$$
regardless of the fact whether $\mathrm{e}^{-\beta v}$ is apriori summable or not. We also prove natural generalizations of this semiclassical limit to a large class of covariant Schrödinger operators that act on sections in Hermitian vector bundle over $(X,m,b)$, a result that particularly applies to magnetic Schrödinger operators that are defined on $(X,m,b)$. 
 \end{abstract}


\setcounter{page}{1}

\date{\today} %
\maketitle

\section{Introduction}

Let us recall some classical facts from the Euclidean $\IR^l$: Assume that $-\Delta+v$ is a Schrödinger operator in $\mathsf{L}^2(\IR^l)$ with $v:\IR^l\to \IR$ (say) bounded from below and $\mathrm{e}^{-\beta v}\in\mathsf{L}^1(\IR^l)$ for some $\beta>0$. Then one has the following behaviour of the corresponding quantum partition function $\mathrm{tr}(\mathrm{e}^{-\beta(-\f{\hbar}{2}\Delta+v})$:
\begin{align}
(2\pi)^l\hbar^{\f{l}{2}}\mathrm{tr}\big(\mathrm{e}^{-\beta(-\f{\hbar}{2}\Delta+v)}\big)\xrightarrow[]{\hbar\to 0+} \int_{\IR^l}\int_{\IR^l}\mathrm{e}^{-\beta\big(\f{|p|^2}{2}+v(q)\big)}\Id p  \Id q.\label{sedq}
\end{align}
Clearly, (\ref{sedq}) is a semiclassical limit, as the integral on the right hand side is over the classical phase space of the system, and thus can be seen as the classical partition function of the system. The analogous result holds true with the same right hand side as in (\ref{sedq}), if one couples $-\Delta$ to a magnetic field. \\
If one thinks about realizing analogous results in a \emph{discrete configuration space}, one is apriori faced with the problem that it is not really clear what the underlying phase space should be. However, if one realizes that (\ref{sedq}) is equivalent to
\begin{align}
(2\pi\hbar)^{\f{l}{2}}\mathrm{tr}\left(\mathrm{e}^{-\beta\hbar\left(-\f{1}{2}\Delta+v/\hbar\right)}\right)\xrightarrow[]{\hbar\to 0+} \int_{\IR^l}\mathrm{e}^{-\beta v(x)}  \Id x,\label{sedq2}
\end{align}
it is clear that the latter problem only exists conceptually, not mathematically. A typical path integral proof \cite{Si05} of (\ref{sedq2}) relies on a Golden-Thompson type inequality of the form
\begin{align}
\label{gold}\mathrm{tr}\left(\mathrm{e}^{-t\left(-\f{1}{2}\Delta +v\right)}\right)&\leq\int_{\IR^l}\mathrm{e}^{ \f{t}{2} \Delta}(x,x)\mathrm{e}^{- t v(x)}  \Id x\>\text{ for all $t>0$,}\\
\nn &\text{which is of course }\>=(2\pi t)^{-\f{l}{2}}\int_{\IR^l}\mathrm{e}^{- t v(x)}  \Id x
\end{align}
and thus shows that the necessity of the regularizing factor $(2\pi\beta\hbar)^{\f{l}{2}}$ in (\ref{sedq2}) ultimately comes from the on-diagonal singularity at $t=0$ of $\mathrm{e}^{ \f{t}{2} \Delta}(\bullet,\bullet)$. \\
On the other hand, it has been shown recently \cite{GKS-13, gmt} that if $H$ is the (generally unbounded) operator corresponding to the canonical Dirichlet form on an arbitrary possibly locally infinite weighted graph\footnote{Here, $X$ is a countable set of vertices, $b:X\times X\to [0,\infty)$ an edge weight function, and $m:X\to (0,\infty)$ a vertex weight function; cf. section \ref{SS:setting} for the precise definitions.} $(X,b,m)$ and if $v:X\to\IR$ is a potential such that $H+v$ is well-defined as a form sum, then one still has a Golden-Thompson inequality 
\begin{align}
\mathrm{tr}\left(\mathrm{e}^{-t\left(H+v\right)}\right)\leq \sum_{x\in X}\mathrm{e}^{-t H}(x,x)\mathrm{e}^{- t v(x)}  m(x)\>\text{ for all $t>0$,}\label{gold2}
\end{align}
the fundamental difference of this discrete setting to the continuum setting is, however, that one has 
$$
\mathrm{e}^{-t H}(x,y)m(x)\leq 1,
$$ 
as $\mathrm{e}^{-t H}(x,y)m(x)$ is precisely the probability of finding the underlying Markoff particle in $y$ at the time $t$, when conditioned to start in $x$. In particular, the right hand side of (\ref{gold2}) is $\leq \sum\mathrm{e}^{- t v}$. At this point, at the latest, it becomes natural to ask whether the quantum partition function $\mathrm{tr}\left(\mathrm{e}^{-\beta\hbar\left(H+v/\hbar\right)}\right)$ satisfies
\begin{align}
\mathrm{tr}\left(\mathrm{e}^{-\beta\hbar\left(H+v/\hbar\right)}\right)\xrightarrow[]{\hbar\to 0+} \sum_{x\in X}\mathrm{e}^{-\beta v(x)}  .\label{sedq3}
\end{align}

One of our main results, \emph{Theorem \ref{main}, states that (\ref{sedq3}) indeed is always true, whenever $v$ is such that $H+v/\hbar$ is well-defined as a form sum for all $\hbar >0$}. Here, in contrast to the continuum setting, we even do not have to assume $\sum\mathrm{e}^{-\beta v}<\infty$, in the sense that the limit in (\ref{sedq3}) exists anyway. This admits entirely geometric results such as
\begin{align}
\mathrm{tr}\left(\mathrm{e}^{-\hbar ( H+\tilde{m}\hbar)}\right)\xrightarrow[]{\hbar\to 0+} \sum_{x\in X}m(x),\>\text{ with }\>\tilde{m}:=-\mathrm{ln}(m):X\longrightarrow \IR,
\end{align}
under very weak assumptions on $m$. \\
Finally, using results from \cite{gmt}, we extend Theorem \ref{main} to covariant Schrödinger operators that act on sections in Hermitian vector bundles over $(X,b,m)$ in Theorem \ref{main2}. These vector-valued operators have found applications in combinatorics, physics and image processing \cite{Si-Vu-12,  Kenyon-11,FKC-S-92}. In particular, Theorem \ref{main2} applies to magnetic Schrödinger operators which have been considered in \cite{Su,MT,Mi,Mi2,MY,DM,MSY}, and more recently in \cite{GKS-13}. This perfectly matches with the magnetic variant of (\ref{sedq}). Our proofs of Theorem \ref{main} and Theorem \ref{main2} are entirely probabilistic, and they rely on various (covariant) Feynman-Kac formulae that have been established recently in \cite{GKS-13,gmt}. \vspace{1.2mm}

This paper is organized as follows: Section \ref{S:main} is devoted to the formulation of the main results Theorem \ref{main} and Theorem \ref{main2}. To this end, we first explain the underlying Dirichlet space setting of weighted discrete graphs and Hermitian vector bundles thereon in Section \ref{SS:setting}. In Section \ref{ops}, we give precise functional analytic definitions of the underlying (covariant) Schrödinger operators, so that finally Theorem \ref{main} and Theorem \ref{main2} can be formulated in Section \ref{crims}. Then Section \ref{bew} is devoted to the proofs of these two theorems. Since, as we have already remarked, these proofs use partially technical probabilistic results from \cite{GKS-13,gmt}, we first completely develop the corresponding notation in Section \ref{bew} for the convenience of the reader. \vspace{4mm}

{\bf Acknowledgements.} The author (B.G.) would like to thank Matthias Keller for a very helpful discussion. B.G. has been financially supported by the SFB 647: Raum-Zeit-Materie.

\section{Main results}\label{S:main}

\subsection{Setting}\label{SS:setting}

Let $(X,b,m)$ be an arbitrary weighted graph, that is, $X$ is a countable set, $b$ is a symmetric function
$$
b: X\times X\longrightarrow  [0,\infty) 
$$
with the property that 
$$
b(x,x)=0,\>\>\sum_{y\in X} b(x,y)<\infty\>\>\text{ for all $x\in X$,}
$$
and $m:X\to (0,\infty)$ is an arbitrary function. Here, $X$ is interpreted as the set of vertices, $\{b>0\}$ is interpreted as the set of (weighted) edges of the graph $(X,b)$, and $m$ is considered as a vertex weight function (see also Remark \ref{lattice} below).\\
Given $x,y\in X$ we write $x\sim_b y$, if and only if $b(x,y)>0$, and $(X,b)$ is called connected, if for all $x,y\in X$ there is a finite chain $x_1,\dots, x_n \in X$ such that $x=x_1$, $y=x_n$ and $x_j\sim_b x_{j+1}$ for all $j=1,\dots, n-1$. As $m$ determines a measure on the discrete space $X$ in the obvious way, we get the corresponding  complex Hilbert space of complex-valued $m$-square-summable functions on $X$, which will be denoted by $\ell^2_m (X)$. We use the conventions
\begin{align*}
 \textrm{deg}_{m,b}(x):=\frac{1}{m(x)}\sum_{y\in X}b(x,y),\>\textrm{deg}_{1,b}(x):=\sum_{y\in X}b(x,y)\>\>\text{ for all $x\in X$,}
\end{align*}
and the spaces of complex-valued and complex-valued finitely supported functions on $X$ will be denoted by $\mathsf{C}(X)$ and $\mathsf{C}_c(X)$, respectively, where of course $\mathsf{C}_c(X)$ is dense in $\ell^2_m (X)$. \vspace{1mm}

\emph{We fix an arbitrary (w.l.o.g.) connected weighted graph $(X,b,m)$.}\vspace{1mm}

Next, we recall that a complex vector bundle $F\to X$ on $X$ with $\mathrm{rank}(F)= \nu\in\IN$ is given by a family $F=\bigsqcup_{x\in X}F_x$ of $\nu$-dimensional complex linear spaces, where then the corresponding space of sections is denoted by
$$
\Gamma(X,F)=\left.\big\{f\right| f:X\to F, f(x)\in F_x\big\}.
$$
If additionally each fiber $F_x$ comes equipped with a complex scalar product $(\bullet,\bullet)_x$, then $F\to X$ is referred to as a Hermitian vector bundle, and the norm and operator norm corresponding to $(\bullet,\bullet)_x$ will be denoted with $|\bullet|_x$. We get the corresponding complex Hilbert space of $\ell^2_m$-sections 
$$
\Gamma_{\ell^2_m}(X,F)=\left.\big\{f\right| f\in \Gamma(X,F), |f|\in \ell^2_m(X)\big\}
$$
with its canonical scalar product
$$
\left\langle f_1,f_2\right\rangle_m:=\sum_{x\in x} (f_1(x),f_2(x))_xm(x)
$$ 
and norm $\left\|f\right\|_m:=\sqrt{\left\langle f,f\right\rangle_m}$, where of course $\left\|\bullet\right\|_m$ will also denote the corresponding operator norm. Note that $\Gamma_{\ell^2_m}(X,F)$ generalizes $\ell^2_m(X)$ in a canonic way (cf. Remark \ref{magi} below), and that the space of finitely supported sections $\Gamma_{\mathsf{C}_c}(X,F)$ is dense in $\Gamma_{\ell^2_m}(X,F)$. For any $x\in X$ we write $\mathbf{1}_{x}:F_x\to F_x$ for the identity operator and $\mathrm{tr}_{x}:\mathrm{End}(F_x)\to \IC$ for the underlying canonical trace. In the above situation, a unitary $b$-connection $\Phi$ on $F\to X$ is an assignment which to any $x\sm y$ assigns an unitary map $\Phi_{x,y}: F_x\to F_y$ with $\Phi_{y,x}=\Phi^{-1}_{x,y}$.\vspace{1mm}

\emph{We fix a Hermitian vector bundle $F\to X$ of rank $\nu\in\IN$, with a unitary $b$-connection $\Phi$ defined on it.}

\subsection{Operators under consideration}\label{ops}

We start by defining a sesqui-linear form $\tilde{Q}_{\Phi,0}$ in $\Gamma_{\ell^2_m}(X,F)$ with domain of definition $\Gamma_{c}(X,F)$ by setting
\begin{align}
\tilde{Q}_{\Phi,0}(f_1,f_2)=&\>\frac{1}{2}\sum_{x\sm y}b(x,y)\big( f_1(x)-\Phi_{y,x} f_1(y),f_2(x)-\Phi_{y,x} f_2(y)\big)_{x}.\nn
\end{align}
The form $\tilde{Q}_{\Phi,0}$ is densely defined, symmetric and nonnegative. Furthermore, one has:

\begin{Lemma}\label{absh} $\tilde{Q}_{\Phi,0}$ is closable.
\end{Lemma}

\begin{proof} It is sufficient to show that the map
\begin{align*}
\Gamma_{\ell^2_m}(X,E)\longrightarrow [0,\infty],\>f\longmapsto \frac{1}{2}\sum_{x\sm y}b(x,y)\big| f(x)-\Phi_{y,x} f(y)\big|^2_{x}
\end{align*}
is lower semicontinuous. However, from the discreteness of the underlying measure space we get the implication
$$
\text{$\left\|f_n-f\right\|_m\xrightarrow[]{n\to \infty}0\>\Rightarrow \>|f_n(z)- f(z)|_z\xrightarrow[]{n\to \infty} 0$ for all $z\in X$,} 
$$
so that the claim follows from Fatou\rq{}s lemma.
\end{proof}

We denote the closure of $\tilde{Q}_{\Phi,0}$ by $Q_{\Phi,0}$ and the self-adjoint operator corresponding to the latter form by $H_{\Phi,0}$. 

\begin{Remark}\label{magi} Of course, we can deal with usual scalar functions on $X$ simply by taking $F_x=\{x\}\times\IC$ with its canonical Hermitian structure. Indeed, then the sections in $F\to X$ can be canonically identified with complex-valued functions on $X$ any $\Phi$ can uniquely be written as $\Phi(x,y)=\mathrm{e}^{\mathrm{i}\theta(x,y)}$ with $\theta$ is a magnetic potential on $(X,b)$, that is, an antisymmetric function 
$$
\theta: \{b>0\}\longrightarrow [-\pi,\pi],
$$
and as a particular case of the above constructions, we get the corresponding forms $\tilde{Q}^{\mathrm{scal}}_{\theta,0}$ and $Q^{\mathrm{scal}}_{\theta,0}$ in $\ell^2_m(X)$. The operator corresponding to $Q^{\mathrm{scal}}_{\theta,0}$ will be denoted with $H^{\mathrm{scal}}_{\theta,0}$. The particular form $Q:=Q^{\mathrm{scal}}_{0,0}$ is a regular Dirichlet form $\ell^2_m(X)$, when $X$ is equipped with its discrete topology. Let  $H:=H^{\mathrm{scal}}_{0,0}$ denote the corresponding operator, and let
\begin{align*}
[0,\infty)\times X\times X\longrightarrow (0,\infty),\>(t,x,y)\longmapsto p(t,x,y)
\end{align*}
be the integral kernel corresponding to 
$$
(\mathrm{e}^{-t H})_{t\geq 0}\subset \ILL(\ell^2_m(X)), 
$$
where $p(\bullet,\bullet,\bullet)>0$ is implied by the connectedness of $(X,b)$. 
\end{Remark}

Let us make the essential point of Lemma \ref{absh} clear: $\tilde{Q}_{\Phi,0}$ is closable although, in general, this form need not come from a well-defined symmetric operator $\Gamma_{\ell^2_m}(X,F)$ with domain $\Gamma_{c}(X,F)$. For a precise statement in this context, set
\begin{align*}
 \Gamma(X,F;b):=\left\{f\left|f\in\Gamma(X,F), \sum_{y\in X}b(x,y)|f(y)|_y<\infty\>\text{ for all $x\in X$}\right\}\right.
\end{align*}
and define a covariant formal difference operator $\tilde{H}_{\Phi,0}$ by
\begin{align*}
&\tilde{H}_{\Phi,0}: \Gamma(X,F;b)\longrightarrow \Gamma(X,F)\\
&\tilde{H}_{\Phi,0}f(x)=\>\frac{1}{m(x)}\sum_{\{y| \> y\sim_b x\}}b(x,y)\big( f(x)-\Phi_{y,x} f(y)\big) ,\>\>x\in X,
\end{align*}
with the obvious notation $\tilde{H}^{\mathrm{scal}}_{\theta,0}$, $\tilde{H}:=\tilde{H}^{\mathrm{scal}}_{0,0}$ in the scalar situation of Remark \ref{magi}.

\begin{Lemma}\label{unn} \emph{a)} If it holds that
\begin{align}
 \sum_{x\in X}\f{b(x,y)^2}{m(x)}<\infty\>\text{  for all $y\in X$,} \label{eins}
\end{align}
then one has
 \begin{align}
&\tilde{H}_{\Phi,0}[\Gamma_c(X,F)]\subset \Gamma_{\ell^2_m}(X,F),\> \Gamma_{\ell^2_m}(X,F) \subset \Gamma(X,F;b)\label{qwq}\\
&\tilde{Q}_{\Phi,0}(f_1,f_2)=\langle \tilde{H}_{\Phi,0} f_1,f_2\rangle_m\>\text{ for all $f_1,f_2\in \Gamma_c(X,F)$.}\label{quas}
\end{align}
\emph{b)} If $(X,b)$ is locally finite, in the sense that for all $y\in X$ there are at most only finitely many $x\in X$ with $b(x,y)>0$, then one has (\ref{eins}) and 
\begin{align}
\Gamma(X,F)=\Gamma(X,F;b).\label{gleich}
\end{align}
\emph{c)} Set
$$
C(b,m):= \sup_{x\in X}\mathrm{deg}_{m,b}(x)=\sup_{x\in X}\left(\frac{1}{m(x)}\sum_{y\in X}b(x,y)\right).
$$
If $C(b,m)<\infty$, then one has (\ref{eins}), and $\tilde{H}_{\Phi,0}$ with domain of definition $\Gamma_c(X,F)$ is bounded with $\left\|\tilde{H}_{\Phi,0}\right\|_m\leq 2 C(b,m)$ and $\overline{\tilde{H}_{\Phi,0}}=H_{\Phi,0}$. If $C(b,m)=\infty$, then one either has
$$
\tilde{H}_{\Phi,0}[\Gamma_c(X,F)]\not\subset \Gamma_{\ell^2_m}(X,F),
$$
or $\tilde{H}_{\Phi,0}$ with domain of definition $\Gamma_c(X,F)$ is unbounded. 
\end{Lemma}

\begin{Remark}\label{lattice} 1. Lemma \ref{unn} a) is optimal in the following sense: If (\ref{eins}) is not satisfied, then one has $\tilde{H}[\mathsf{C}_c(X)]\not\subset \ell^2_m(X)$ by Theorem 3.3 in \cite{Keller-Lenz-09}. \\
2. Note that in the scalar situation one has
\begin{align*}
&\tilde{H}^{\mathrm{scal}}_{\theta,0}:\mathsf{C}(X;b)\longrightarrow\mathsf{C}(X)\\
&\tilde{H}^{\mathrm{scal}}_{\theta,0}f(x)=\>\frac{1}{m(x)}\sum_{ \{y|\>y\sim_b x\}}b(x,y)\big( f(x)-\mathrm{e}^{\mathrm{i}\theta(y,x)} f(y)\big),\>\>x\in X.
\end{align*}
3. More specifically, upon taking $X=\IZ^l$, 
$$
b(x,y):=\begin{cases}1, \text{ if $|y-x|_{\IR^l}=1$,}\\0,\text{ else}\end{cases}
$$
$m\equiv 1$, we get the magnetic lattice Laplacian 
\begin{align}
\tilde{H}^{\mathrm{scal}}_{\theta,0} f(x)= \sum_{ \left\{y\left|\>|y-x|_{\IR^l}=1\right\}\right.  }\big( f(x)-\mathrm{e}^{\mathrm{i}\theta(y,x)} f(y)\big),\>\>x\in \IZ^l,\label{mag}
\end{align}
which is bounded in the sense of Lemma \ref{unn} c), with $\overline{\tilde{H}^{\mathrm{scal}}_{\theta,0}}=H^{\mathrm{scal}}_{\theta,0}$. Let us point out two recent papers in this context that deal with the spectral theory of operators that generalize (\ref{mag}) for $\theta=0$ into two directions: Firstly, in \cite{jo} the authors consider $\delta$-perturbations of $\Psi\left(H\right )$ with $\Psi\in\mathsf{C}^1(0,\infty)$ such that $\Psi\rq{}>0$, and in \cite{yev} the authors replace $\IZ^l$ by an arbitrary periodic lattice and perturb the resulting operator with periodic potentials. 
\end{Remark}

\begin{proof}[Proof of Lemma \ref{unn}.] a) By Cauchy-Schwarz we get that (\ref{eins}) implies the second inclusion in (\ref{qwq}), and Green\rq{}s formula (cf. Lemma 3.1 in \cite{MilTu}) shows that the first inclusion in (\ref{qwq}) implies (\ref{quas}). It remains to prove that (\ref{eins}) implies the first inclusion in (\ref{qwq}). To this end, let $f\in \Gamma_c(X,F)$. By writing $f$ as a finite sum $f=\sum_z 1_{\{z\}} f(z)$, we can assume $f=0$ away from some $z\in X$. Then we have
\begin{align*}
&\left\|\tilde{H}_{\Phi,0} f\right\|^2_m=\sum_{x\in X}\f{1}{m(x)}\left|\sum_{\{y| \> y\sim_b x\}}b(x,y)f(x)-\sum_{\{y| \> y\sim_b x\}}b(x,y)\Phi_{y,x}f(y)\right|^2_x\nn\\
&\leq  \sum_{x\in X}\f{1}{m(x)}\left(\sum_{ y\in X}b(x,y)\left|f(x)\right|_x\right)^2 +\sum_{x\in X}\f{1}{m(x)}\left(\sum_{y\in X}b(x,y)\left|f(y)\right|_y\right)^2\nn\\
&\leq \f{|f(z)|^2_{z}}{m(z)}\left(\sum_{ y\in X}b(y,z)\right)^2+|f(z)|^2_z\sum_{x\in X}\f{b(x,z)^2}{m(x)}<\infty.
\end{align*}
b) This is obvious.\\
c) Let $f\in \Gamma_{c}(X,F)$ be arbitrary, and assume first that $C(b,m)<\infty$. Then we have
$$
 \sum_{x\in X}\f{b(x,y)^2}{m(x)}\leq C(b,m) \sum_{x\in X} b(x,y)  <\infty,
$$
in particular, $\tilde{H}_{\Phi,0}$ is a well-defined linear operator in $\Gamma_{\ell^2_m}(X,F)$ with domain of definition $\Gamma_{c}(X,F)$, and  
\begin{align*}
&\langle \tilde{H}_{\Phi,0} f,f\rangle_m=\tilde{Q}_{\Phi,0}(f) \leq 2\sum_{x,y\in X} b(x,y)|f(x)|^2_x \leq 2 C(b,m) \left\|f\right\|_m^2
\end{align*}
which entails that $\tilde{H}_{\Phi,0}$ is bounded with operator norm $\leq 2C(b,m)$ (see also \cite{davies} for the particular case $\tilde{H}_{\Phi,0}=\tilde{H}$). The assertion $\overline{\tilde{H}_{\Phi,0}}=H_{\Phi,0}$ follows from $\langle \tilde{H}_{\Phi,0} f,f\rangle_m=\tilde{Q}_{\Phi,0}(f)$. \\
Assume now that $C(b,m)=\infty$, and that
$$
\tilde{H}_{\Phi,0}[\Gamma_c(X,F)]\subset \Gamma_{\ell^2_m}(X,F).
$$
Then $\tilde{Q}$ is unbounded from above by Theorem 8.1 in \cite{kellenz}. But in view of
\begin{align*}
 \langle \tilde{H}_{\Phi,0} f,f\rangle_m=\tilde{Q}_{\Phi,0}(f)\geq  \tilde{Q}(|f|),
\end{align*}
where we have used Green's formula and 
$$
\big|f(x)-\Phi_{y,x}f (y)\big|^2_y\geq \big||f(x)|_x-|f (y)|_y\big|^2,
$$
we find that $\tilde{H}_{\Phi,0}$ is unbounded as well.
\end{proof}

Let us now take perturbations by potentials into account. Given a potential $V$ on $F\to X$, that is, $V\in\Gamma(X,\mathrm{End}(F))$ is pointwise self-adjoint, we can define a symmetric sesqui-linear form $Q_V$ in $\Gamma_{\ell^2_m}(X,F)$ by
$$
Q_V(f_1,f_2)=\sum_{x\in X} \big( V(x)f_1(x),f_2(x)\big)_{x}m(x)
$$
on its maximal domain of definition. Whenever $V$ as above admits a decomposition $V=V^{+}-V^{-}$ into potentials $V^{\pm}\geq 0$ (fiberwise as self-adjoint operators) such that $Q_{V^{-}}$ is $Q_{\Phi,0}$-bounded with bound $<1$, that is, if there are $0<C_1<1$, $C_2>0$ such that
$$
Q_{V^-}(f)\leq C_1 Q_{\Phi,0}(f)+C_2 \left\|f\right\|_m^2\text{ for all $f\in\mathsf{D}(Q_{\Phi,0})$,}
$$
then the sesqui-linear form $Q_{\Phi,V}:=Q_{\Phi,0}+Q_{V}$ is densely defined, symmetric, closed and semi-bounded (from below). The self-adjoint semi-bounded operator corresponding to $Q_{\Phi,V}$ will be denoted with $H_{\Phi,V}$. 

\begin{Remark} Note that in the scalar situation of Remark \ref{magi}, any potential can be identified with a function $v:X\to \IR$, and as a particular of the above construction we get the quadratic form $Q^{\mathrm{scal}}_{\theta,v}$ in $\ell^2_m(X)$. The corresponding operator will be written as $H^{\mathrm{scal}}_{\theta,v}$.
\end{Remark}

One has:

\begin{Propandef} \emph{a)} If a potential $V$ on $F\to X$ admits a decomposition $V=V^{+}-V^{-}$ into potentials $V^{\pm}\geq 0$ such that $Q_{|V^{-}|}$ is $Q$-bounded with bound $<1$, then $Q_{V^{-}}$ is also $Q_{\Phi,0}$-bounded with bound $<1$, in particular, $H_{\Phi,V}$ is well-defined. In this case, we say that \emph{$V$ is $Q$-decomposable}, and that $V=V^{+}-V^{-}$ \emph{is a $Q$-decomposition of $V$}.\\
\emph{b)} If a potential $V$ on $F\to X$ admits a decomposition $V=V^{+}-V^{-}$ into potentials $V^{\pm}\geq 0$ such that $Q_{|V^{-}|}$ is infinitesimally $Q$-bounded, that is, if for any $\varepsilon>0$ there is a $C(\varepsilon)>0$ such that
$$
Q_{|V^-|}(f)\leq \varepsilon Q(f)+C(\varepsilon) \left\|f\right\|_m^2\text{ for all $f\in\mathsf{D}(Q)$,}
$$
then $Q_{V^{-}}$ is also infinitesimally $Q_{\Phi,0}$-bounded, in particular, $H_{\Phi,V}$ is well-defined. In this case, we say that \emph{$V$ is infinitesimally $Q$-decomposable}, and that $V=V^{+}-V^{-}$ \emph{is an infinitesimal $Q$-decomposition of $V$}.
\end{Propandef}

\begin{proof} Both statements follow from the fact that for any $f\in\mathsf{D}(Q_{\Phi,0})$ one has  $|f|\in\mathsf{D}(Q)$ with $Q_{\Phi,0}(f)\geq Q(|f|)$ (cf. Theorem 2 (ii) in \cite{gmt}).
\end{proof}

\begin{Remark}\label{fdg} 1. If a potential $V$ on $F\to X$ admits a decomposition $V=V^{+}-V^{-}$ into potentials $V^{\pm}\geq 0$ such that $|V^{-}|\in \mathcal{K}(Q)$, the Kato class of $Q$ (for example if $V\geq -C$ for some constant $C>0$), then $V=V^{+}-V^{-}$ is an infinitesimal $Q$-decomposition of $V$ (cf. Theorem 3.1 in \cite{peter}). Recall here that $w:X\to \IC$ is in $\mathcal{K}(Q)$, if and only if
$$
\lim_{t\to 0+}\sup_{x\in X}\int^t_0\sum_{y\in X} p(s,x,y) |w(y)| m(y)\Id s =0.
$$

2. If $V=V^{+}-V^{-}$ is an infinitesimal $Q$-decomposition of a potential $V$ on $F\to X$, then $V/\hbar=V^{+}/\hbar-V^{-}/\hbar$ is an infinitesimal $Q$-decomposition of $V/\hbar$ for any $\hbar>0$, in particular, $H_{\Phi,V/\hbar}$ is well-defined.
\end{Remark}

We refer the reader to \cite{MilTu} for problems concerning the explicit domain of definition of $H_{\Phi,V}$, and essential-selfadjointness questions related with $H_{\Phi,V}$.

\subsection{The semiclassical limit of the quantum partition function}\label{crims}

We can now formulate the main results of this paper. Firstly, we consider the scalar nonmagnetic situation:

\begin{Theorem}\label{main} Assume that $w:X\to\IR$ is infinitesimally $Q$-decomposable. Then for all $\beta>0$ one has
\begin{align}
&\mathrm{tr}\big(\mathrm{e}^{-\beta\hbar  H^{\mathrm{scal}}_{0,w/\hbar}}\big)\leq \sum_{x\in X}\mathrm{e}^{-\beta w(x)}\in [0,\infty]\>\>\text{ for all $\hbar >0$},\label{lim0}\\
&\liminf_{\hbar \to 0+}\mathrm{tr}\big(\mathrm{e}^{-\beta\hbar H^{\mathrm{scal}}_{0,w/\hbar}}\big) \geq \sum_{x\in X}\mathrm{e}^{-\beta w(x)}, \label{lim1}
\end{align}
in particular, 
\begin{align}
\mathrm{tr}\big(\mathrm{e}^{-\beta\hbar H^{\mathrm{scal}}_{0,w/\hbar}}\big) \xrightarrow[]{\hbar\to 0+} \sum_{x\in X}\mathrm{e}^{-\beta w(x)}. \label{lim}
\end{align}
\end{Theorem}

Note that we do not assume anything on $(X,b,m)$ here, and moreover, that we do not assume anything on $v$ other than $H^{\mathrm{scal}}_{0,w}$ be well-defined: Indeed, as the formulation of Theorem \ref{main} indicates, the limit of $\mathrm{tr}\big(\mathrm{e}^{-\beta\hbar H^{\mathrm{scal}}_{0,w/\hbar}}\big)$ as $\hbar\to 0+$ always exists, regardless of the fact whether or not one has $\sum\mathrm{e}^{-\beta w}<\infty$. This existence heavily requires that no \lq\lq{}magnetic effects\rq\rq{} are present, which admits the usage of arguments that rely on positivity preservation. \\
We immediately get the following Weyl-type geometric result:

\begin{Corollary}\label{mio} If the function $\tilde{m}:X\to \IR$, $x\mapsto -\mathrm{ln}(m(x))$ is infinitesimally $Q$-decomposable (for example, if $m\leq C$ for a constant $C>0$), then one has
$$
\mathrm{tr}\big(\mathrm{e}^{-\hbar H^{\mathrm{scal}}_{0,\tilde{m}/\hbar}}\big) \xrightarrow[]{\hbar\to 0+} \sum_{x\in X}m(x)\in [0,\infty].
$$
\end{Corollary}

We return to the general vector bundle setting. Here, the following can be said:

\begin{Theorem}\label{main2} Assume that the potential $V$ on $F\to X$ is infinitesimally $Q$-decomposable, and that there is a $Q$-decomposable function $w:X\to\IR$ with $V\geq w$. Then for all $\beta>0$ with 
$$
\sum_{x\in X}\mathrm{e}^{-\beta w(x)}<\infty 
$$
one has
\begin{align}
&\mathrm{tr}(\mathrm{e}^{-\beta\hbar H_{\Phi,V/\hbar}}) \leq  \sum_{x\in X}\mathrm{tr}_x(\mathrm{e}^{-\beta V(x)})<\infty, \label{lim2}\\
&\mathrm{tr}(\mathrm{e}^{-\hbar\beta H_{\Phi,V/\hbar}}) \xrightarrow[]{\hbar\to 0+} \sum_{x\in X}\mathrm{tr}_x(\mathrm{e}^{-\beta V(x)}).\label{lim3}
\end{align}
\end{Theorem}

Again, we do not assume anything on $(X,b,m)$ here, but now we cannot use monotonicity arguments to guarantee the inequality in (\ref{lim2}), which requries the lower spectral bound $V\geq w$ for some suitable $w$ with $\sum\mathrm{e}^{-\beta w}<\infty$.

\begin{Remark} If $V$ is an infinitesimally $Q$-decomposable potential on $F\to X$, then 
$$
w:=\min\mathrm{spec}(V(\bullet)):X\longrightarrow \IR
$$
is easily seen to define an infinitesimally $Q$-decomposable function with $V\geq w$, so that this part of the required control on $V$ is not restrictive.
\end{Remark}

In the case of scalar magnetic operators, Theorem \ref{main2} can be written in the following compact form:

\begin{Corollary} Let 
$$
\theta: \{b>0\}\longrightarrow  [-\pi,\pi]
$$ 
be a magnetic potential, and let $w:X\to\IR$ be infinitesimally $Q$-decomposable. Then for all $\beta>0$ with 
$$
\sum_{x\in X}\mathrm{e}^{-\beta w(x)}<\infty
$$
one has
\begin{align*}
\mathrm{tr}\big(\mathrm{e}^{-\beta\hbar H^{\mathrm{scal}}_{\theta,w/\hbar}}\big) \leq  \sum_{x\in X}\mathrm{e}^{-\beta w(x)}<\infty,\>\mathrm{tr}\big(\mathrm{e}^{-\beta\hbar H^{\mathrm{scal}}_{\theta,w/\hbar}}\big) \xrightarrow[]{\hbar\to 0+} \sum_{x\in X}\mathrm{e}^{-\beta w(x)}.
\end{align*}
\end{Corollary}

\section{Proof of Theorem \ref{main} and Theorem \ref{main2}}\label{bew}

We start by recalling some probabilistic facts from \cite{GKS-13,gmt}: Let $(\Omega,\mathscr{F},\mathbb{P})$ be a fixed probability space, let $(Y_n)_{n\in\mathbb{N}_0}$ be a discrete time Markov chain with values in $X$ such that
\begin{align}
\mathbb{P}(Y_n=x|Y_{n-1}=y)=\frac{b(x,y)}{\textrm{deg}_{1,b}(x)}\qquad \textrm{for all }n\in\IN.\label{trans}
\end{align}
Let $(\xi_n)_{n\in\mathbb{N}_{0}}$ be a sequence of independent exponentially distributed random variables with parameter $1$, which are independent of $(Y_n)_{n\in\mathbb{N}_{0}}$. For $n\in\IN$ set
\[
S_n:=\frac{1}{\textrm{deg}_{m,b}(Y_{n-1})}\xi_n,\qquad \tau_n:=S_1+S_2+\dots+S_n,
\]
where $\tau_0:=0$, so that we get the predictable stopping time
$$
\tau:=\sup\{\tau_n|n\in\mathbb{N}_{0} \}  >0,
$$
which allows us to define the maximally defined, right-continuous process
$$
\mathbb{X}\colon [0,\tau)\times \Omega\longrightarrow  X,\>\>
\mathbb{X}|_{[\tau_n,\tau_{n+1})\times \Omega}:=Y_n,\>\>n\in\mathbb{N}_{0},
$$
which has the $\tau_n$\rq{}s as its jump times and the $S_n$\rq{}s as its holding times. Ultimately, with $\mathbb{P}^{x}:=\mathbb{P}(\bullet| \mathbb{X}_{0}=x)$ and $\IFF_*=(\IFF_t)_{t\geq 0}$ the natural filtration of $\IFF$ given by $\mathbb{X}$, it turns out that $(\Omega, \IFF_*, \mathbb{X}, (\mathbb{P}^x)_{x\in X})$ is a reversible strong Markov process. Let us introduce the process
$$
N(t):=\sup\{n\in\IN_0 | \tau_n\leq t\}:\Omega\longrightarrow \IN_0\cup \{\infty\},
$$
which at a fixed time $t\geq 0$ counts the jumps of $\mathbb{X}$ until $t$, so that
\begin{align}
\{N(t)<\infty\}=\{t<\tau\}\>\text{ for all $t\geq 0$.}\label{bez}
\end{align}

Then we have:

\begin{Lemma} The following statements hold for all $t\geq 0$, $x,y\in X$,
\begin{align}
&\mathbb{P}\big(b(\mathbb{X}_{\tau_n},\mathbb{X}_{\tau_{n+1}}) >0\text{\emph{ for all $n\in\IN_0$}}\big) =1,\label{jump}\\
&p(t,x,y)m(y)=\mathbb{P}^x(\mathbb{X}_{t}=y),\label{deg1}\\
&\mathbb{P}^{x}(N(t)=0) = \mathrm{e}^{-\mathrm{deg}_{m,b}(x)t}.\label{deg}
\end{align}
\end{Lemma}

\begin{proof} The relation (\ref{jump}), which simply means that $\mathbb{X}$ can only jump to neighbors, follows immediately from the definitions, (\ref{deg1}) is well-known (and follows from example from the Feynman-Kac formula for $\mathrm{e}^{-tH}$ \cite{GKS-13}), and (\ref{deg}) follows from
$$
\mathbb{P}^{x}(N(t)=0) =\mathbb{P}^{x}(t<\tau_1)=\mathbb{P}^{x}(\mathrm{deg}_{m,b}(x)t<\zeta_1) =\mathrm{e}^{-\mathrm{deg}_{m,b}(x)t},
$$
which follows immediately from the definitions.
 \end{proof}
 
For $x,y\in X$ and $t>0$ let $\mathbb{P}^{x,y}_t$ be the probability measure on $\{t<\tau\}$ given by $\mathbb{P}^{x,y}_t=\mathbb{P}^{x}(\bullet|\mathbb{X}_t=y)$. With these preparations, we can state the following result, which the proof of Theorem \ref{main} relies on:

\begin{Theorem}\label{sed} Assume that $\tilde{w}:X\to\IR$ is $Q$-decomposable. Then the following assertions hold true:

\emph{a)} For any $t>0$ one has 
\begin{align}
\mathrm{tr}\big(\mathrm{e}^{-t H^{\mathrm{scal}}_{0,\tilde{w}}}\big)=\sum_{x\in X}p(t,x,x) \mathbb{E}^{x,x}_t\left[\mathrm{e}^{-\int^t_0 \tilde{w}(\mathbb{X}_s)\Id s}\right]m(x)\in [0,\infty].\label{for}
\end{align}

\emph{b)} For any $t>0$ one has
\begin{align}
\mathrm{tr}\big(\mathrm{e}^{-t H^{\mathrm{scal}}_{0,\tilde{w}}}\big)\leq \sum_{x\in X}\mathrm{e}^{-t\tilde{w}(x)} \in [0,\infty].\label{abs}
\end{align}
\end{Theorem}

\begin{proof} This is precisely Theorem 5.3 in \cite{GKS-13}.
\end{proof}

Here, it should be noted that the proof of the Golden-Thompson inequality (\ref{abs}) is highly nontrivial. In particular, (\ref{abs}) does not follow directly from (\ref{for}). Now we can give the actual proof of Theorem \ref{main}: 

\begin{proof}[Proof of Theorem \ref{main}]
Applying (\ref{abs}) with $\tilde{w}=w/\hbar$ (keeping Remark \ref{fdg}.2 in mind) and $t=\beta\hbar$ and using (\ref{deg1}) immediately implies
$$
\mathrm{tr}\big(\mathrm{e}^{-\beta\hbar H^{\mathrm{scal}}_{0,w/\hbar}}\big)\leq \sum_{x\in X}\mathrm{e}^{-\beta w(x)}.
$$
In order to see that $\liminf_{\hbar\to 0+}\mathrm{tr}(\dots)$ is bounded from below by $\sum\mathrm{e}^{-\beta w}$, we remark that applying (\ref{for}) with $\tilde{w}=w/\hbar$ and $t= \beta\hbar$ implies the first equality in
\begin{align}
\mathrm{tr}\big(\mathrm{e}^{-\beta\hbar H^{\mathrm{scal}}_{0,w/\hbar}}\big)&=\sum_{x\in X}\mathbb{P}^{x}(\mathbb{X}_{\beta\hbar}=x) \mathbb{E}^{x,x}_{\beta\hbar}\left[\mathrm{e}^{-\f{1}{\hbar}\int^{\beta\hbar}_0 w(\mathbb{X}_s)\Id s}\right]\nn\\
&\geq \sum_{x\in X}\mathbb{P}^{x}(\mathbb{X}_{\beta\hbar}=x) \mathbb{E}^{x,x}_{\beta\hbar}\left[1_{\{N(\beta\hbar)=0\}}\mathrm{e}^{-\f{1}{\hbar}\int^{\beta\hbar}_0 w(\mathbb{X}_s)\Id s}\right]\nn\\
&= \sum_{x\in X}\mathbb{P}^{x}(\mathbb{X}_{\beta\hbar}=x) \mathbb{P}^{x,x}_{\beta\hbar}(N(\beta\hbar)=0)\mathrm{e}^{-\beta w(x)}\nn\\
&=\sum_{x\in X}\mathbb{P}^{x}(\mathbb{X}_{\beta\hbar}=x) \f{\mathbb{P}^{x}(N(\beta\hbar)=0, \mathbb{X}_{\beta\hbar}=x)}{\mathbb{P}^{x}(\mathbb{X}_{\beta\hbar}=x)}\mathrm{e}^{-\beta w(x)}\nn\\
&=\sum_{x\in X}\mathbb{P}^{x}(N(\beta\hbar)=0)\mathrm{e}^{-\beta w(x)}\nn\\
&= \sum_{x\in X} \mathrm{e}^{-\mathrm{deg}_{m,b}(x)\beta\hbar}\mathrm{e}^{-\beta w(x)},\nn
\end{align}
where we have used (\ref{deg}) for the last equality. Thus, from Fatou\rq{}s lemma we can conclude
$$
\liminf_{\hbar\to 0+}\mathrm{tr}\big(\mathrm{e}^{-\beta\hbar H^{\mathrm{scal}}_{0,w/\hbar}}\big)\geq \sum_{x\in X}\mathrm{e}^{-\beta w(x)},
$$
which completes the proof.

\end{proof}

Before we come to the proof of Theorem \ref{main2}, we first have to explain the additional ingredients of the Feynman-Kac formula for $\mathrm{e}^{-t H_{\Phi,V}}$: Firstly, using (\ref{bez}) and (\ref{jump}), the $\Phi$-parallel transport along the paths of $\mathbb{X}$ is well-defined by 
\begin{align*}
&\pa^{\Phi}: [0,\tau)\times \Omega\longrightarrow F \boxtimes F^* =\bigsqcup_{(x,y)\in X\times X}\mathrm{Hom}(F_y,F_x)\\
&\pa^{\Phi}_t:=
\begin{cases}
\mathbf{1}_{\mathbb{X}_0},\> \text{ if $N(t)=0$}\\ \\
 \Phi_{\mathbb{X}_{\tau_{N(t)-1}},\mathbb{X}_{\tau_{N(t)}}}\cdots \Phi_{\mathbb{X}_{\tau_{0}},\mathbb{X}_{\tau_{1}}}\>\text{ else}
\end{cases}\in \mathrm{Hom}(F_{\mathbb{X}_0},F_{\mathbb{X}_t}),
\end{align*}
which gives a pathwise unitary process, and furthermore we have the process 
$$
\mathscr{A}^{\Phi,V}:[0,\tau)\times\Omega\longrightarrow \mathrm{End}(F)
$$
which is given by the time ordered exponential
\begin{align*}
&\mathscr{A}_{t}^{\Phi,V}-\mathbf{1}_{\mathbb{X}_0}\\
&=\sum_{n=1}^{\infty}(-1)^n\int_{t\sigma_n}\pa^{\Phi,-1}_{s_1}V(\mathbb{X}_{s_1})\pa^{\Phi}_{s_1}\cdots \pa^{\Phi,-1}_{s_1}V(\mathbb{X}_{s_1})\pa^{\Phi}_{s_1} \,\Id s_1\,\cdots\,\Id s_n\\
&\in\mathrm{End}(F_{\mathbb{X}_0}) ,
\end{align*}
where
\[
t\sigma_n:=\big\{(s_1,s_2,\dots,s_n)| \>0\leq s_1\leq s_2\leq\cdots\leq s_n\leq t\big\}\subset \IR^n
\]
denotes the $t$-scaled standard $n$-simplex, and where the series converges pathwise absolutely and locally uniformly in $t$. With these preparations, we have the following covariant anaologue of Theorem \ref{sed}:

\begin{Theorem} Assume that the potential $\tilde{V}$ on $F\to X$ is $Q$-decomposable. Then the following assertions hold true:

\emph{a)} For any $t>0$ one has 
\begin{align}
\mathrm{tr}(\mathrm{e}^{-t H_{\Phi,\tilde{V}}})=\sum_{x\in X}p(t,x,x) \mathbb{E}^{x,x}_t\left[\mathscr{A}^{\Phi,\tilde{V}}_t\pa^{\Phi,-1}_t\right]m(x)\in [0,\infty].\label{for2}
\end{align}

\emph{b)} If there is a $Q$-decomposable function $\tilde{w}:X\to\IR$ with $\tilde{V}\geq \tilde{w}$ and 
$$
\sum_{x\in X}\mathrm{e}^{-\tilde{w}(x)}<\infty, 
$$
then for any $t>0$ one has
\begin{align}
\mathrm{tr}(\mathrm{e}^{-t H_{\Phi,\tilde{V}}})\leq \sum_{x\in X}\mathrm{tr}_x(\mathrm{e}^{-t\tilde{V}(x)}) <\infty.\label{abs2}
\end{align}
\end{Theorem}

\begin{proof} Noting that $p(t,x,x)m(x)\leq 1$ by (\ref{deg1}), these results follow immediately from Theorem 4 in \cite{gmt}. 
\end{proof}

Now we are prepared for the proof of Theorem \ref{main2}:

\begin{proof}[Proof of Theorem \ref{main2}] Applying (\ref{abs2}) with $\tilde{V}=V/\hbar$, $\tilde{w}=w/\hbar$ (keeping again Remark \ref{fdg}.2 in mind) and $t=\beta\hbar$ directly gives (\ref{lim2}). In order to see (\ref{lim3}), we apply (\ref{for2}) with $\tilde{V}=V/\hbar$ and $t=\beta\hbar$ to get the first equality in
\begin{align}
\nn\mathrm{tr}(\mathrm{e}^{-\beta\hbar H_{\Phi,V/\hbar}})&=\sum_{x\in X}p(\beta\hbar,x,x) \mathbb{E}^{x,x}_{\beta\hbar}\left[\mathrm{tr}_x\big(\mathscr{A}^{\Phi,V/\hbar}_{\beta\hbar}\pa^{\Phi,-1}_{\beta\hbar}\big)\right]m(x)\\ 
\nn&=\sum_{x\in X}\mathbb{P}^x(\mathbb{X}_{\beta\hbar}=x) \mathbb{E}^{x,x}_{\beta\hbar}\left[\mathrm{tr}_x\big(\mathscr{A}^{\Phi,V/\hbar}_{\beta\hbar}\pa^{\Phi,-1}_{\beta\hbar}\big)\right]\\
\nn&=\sum_{x\in X}\mathbb{E}^{x}\left[1_{\{\mathbb{X}_{\beta\hbar}=x\}}\mathrm{tr}_x\big(\mathscr{A}^{\Phi,V/\hbar}_{\beta\hbar}\pa^{\Phi,-1}_{\beta\hbar}\big)\right]\\
\label{zerl}&=\sum_{x\in X}\mathbb{E}^{x}\left[1_{\{N(\beta\hbar)=0\}}\mathrm{tr}_x\big(\mathscr{A}^{\Phi,V/\hbar}_{\beta\hbar}\pa^{\Phi,-1}_{\beta\hbar}\big)\right]\\
\nn&\>\>\>+\sum_{x\in X}\mathbb{E}^{x}\left[1_{\{N(\beta\hbar)\geq 1\}}1_{\{\mathbb{X}_{\beta\hbar}=x\}}\mathrm{tr}_x\big(\mathscr{A}^{\Phi,V/\hbar}_{\beta\hbar}\pa^{\Phi,-1}_{\beta\hbar}\big)\right]
\end{align}
In view of
\begin{align*}
\mathbb{E}^{x}\left[1_{\{N(\beta\hbar)=0\}}\mathrm{tr}_x\big(\mathscr{A}^{\Phi,V/\hbar}_{\beta\hbar}\pa^{\Phi,-1}_{\beta\hbar}\big)\right]
=\mathbb{P}^{x}(N(\beta\hbar)=0)\mathrm{tr}_x(\mathrm{e}^{-\beta V(x)}),
\end{align*}
the first summand in (\ref{zerl}) tends to $\sum\mathrm{tr}_{\bullet}(\mathrm{e}^{-\beta V(\bullet)})$ as $\hbar \to 0+$, using (\ref{deg}) in combination with $\sum\mathrm{tr}_{\bullet}(\mathrm{e}^{-\beta V(\bullet)})<\infty$ and dominated convergence.\\
Finally, the second summand in (\ref{zerl}) tends to $0$ as $\hbar \to 0+$: To see this, recalling $\nu=\mathrm{rank}(F)$ and using Proposition 6 (i) in \cite{gmt} (a Gronwall type upper bound on the operator norm of $\mathscr{A}^{\Phi,V/\hbar}_{\beta\hbar}$) we can estimate as follows,
\begin{align*}
&\sum_{x\in X}\mathbb{E}^{x}\left[1_{\{N(\beta\hbar)\geq 1\}}1_{\{\mathbb{X}_{\beta\hbar}=x\}}\left|\mathrm{tr}_x\big(\mathscr{A}^{\Phi,V/\hbar}_{\beta\hbar}\pa^{\Phi,-1}_{\beta\hbar}\big)\right|\right]\\
\leq \>&\nu\sum_{x\in X}\mathbb{E}^{x}\left[1_{\{N(\beta\hbar)\geq 1\}}1_{\{\mathbb{X}_{\beta\hbar}=x\}}\left|\mathscr{A}^{\Phi,V/\hbar}_{\beta\hbar}\right|_x\right]\\
\leq \>& \nu \sum_{x\in X}\mathbb{E}^{x}\left[1_{\{N(\beta\hbar)\geq 1\}}1_{\{\mathbb{X}_{\beta\hbar}=x\}}\mathrm{e}^{-\f{1}{\hbar}\int^{\beta\hbar}_0  w(\mathbb{X}_s)\Id s}\right]\\
=\>& \nu \sum_{x\in X}\mathbb{E}^{x}\left[1_{\{\mathbb{X}_{\beta\hbar}=x\}}\mathrm{e}^{-\f{1}{\hbar}\int^{\beta\hbar}_0 w(\mathbb{X}_s)\Id s}\right]\\
&- \nu \sum_{x\in X}\mathbb{E}^{x}\left[1_{\{N(\beta\hbar)=0\}}\mathrm{e}^{-\f{1}{\hbar}\int^{\beta\hbar}_0 w(\mathbb{X}_s)\Id s}\right]\\
= \>& \nu \ \mathrm{tr}\big(\mathrm{e}^{-\hbar \beta H^{\mathrm{scal}}_{0,w/\hbar}}\big)-  \nu \sum_{x\in X}\mathrm{e}^{-\mathrm{deg}_{m,b}(x)\beta\hbar}\mathrm{e}^{- \beta w(x)},\\
\end{align*}
which goes to zero by Theorem \ref{main}, as the second summand goes to $\nu\sum\mathrm{e}^{- \beta w}$ by dominated convergence.  

\end{proof}

\end{document}